\documentclass[conference,a4paper]{IEEEtran}

\usepackage{amsmath, amssymb, cite}
\usepackage{mathtools}
\usepackage{amsthm,multirow,color,amsfonts}
\usepackage{tabulary}
\usepackage{subfigure}
\usepackage{graphicx}
\usepackage{enumerate}

\title{Optimizing The Spatial Content Caching Distribution for Device-to-Device Communications}

\author{
    \IEEEauthorblockN{Derya~Malak\IEEEauthorrefmark{1}, Mazin Al-Shalash\IEEEauthorrefmark{2} and Jeffrey~G.~Andrews\IEEEauthorrefmark{1}}
    \IEEEauthorblockA{\IEEEauthorrefmark{1}Department of Electrical and Computer Engineering \\The University of Texas at Austin, Austin, TX 78701, USA}
    \IEEEauthorblockA{\IEEEauthorrefmark{2}Huawei Technologies, Plano, TX 75075, USA
    \\ Email: deryamalak@utexas.edu, mshalash@huawei.com, jandrews@ece.utexas.edu}
    }

\newtheorem{theo}{Theorem}
\newtheorem{defi}{Definition}

\newtheorem{prop}{Proposition}

\DeclareMathOperator*{\Rdd}{R_{\rm D2D}}
\DeclareMathOperator*{\Rdds}{R_{\rm D2D}^2}
\DeclareMathOperator*{\Pois}{\rm Poisson}
\DeclareMathOperator*{\Phit}{P_{\rm Hit}}

\DeclareMathOperator*{\BM}{\rm BM}
\DeclareMathOperator*{\PPP}{\rm PPP}

\DeclareMathOperator*{\mhc}{\rm MHC}
\DeclareMathOperator*{\DD}{\rm D2D}

\DeclareMathOperator*{\Poisson}{\rm Poisson}

\begin{document}

\maketitle
\begin{abstract}
We study the optimal geographic content placement problem for device-to-device ($\DD$) networks in which the content popularity follows the Zipf law. We consider a $\DD$ caching model where the locations of the $\DD$ users (caches) are modeled by a Poisson point process ($\PPP$) and have limited communication range and finite storage.  Unlike most related work which assumes independent placement of content, and does not capture the locations of the users, we model the spatial properties of the network including spatial correlation in terms of the cached content. We propose two novel spatial correlation models, the exchangeable content model and a Mat\'{e}rn ($\mhc$) content placement model, and analyze and optimize the \emph{hit probability}, which is the probability of a given $\DD$ node finding a desired file at another node within its communication range.  We contrast these results to the independent placement model, and show that exchangeable placement performs worse. On the other hand, $\mhc$ placement yields a higher cache hit probability than independent placement for small cache sizes.
\end{abstract}

\maketitle


\section{Introduction}
\label{intro}

$\DD$ communication is a promising technique for enabling proximity-based applications and increased offloading from the heavily loaded cellular network, and is being actively standardized by 3GPP \cite{LinMag2014}.  Its efficacy however, requires nearby users to possess content that another user wants.  Therefore, intelligent caching of popular files is indispensable for $\DD$ to be successful.  Caching has been shown to provide increased spectral reuse and throughput gain in $\DD$-enabled networks \cite{Naderializadeh2014}, but the optimal way to spatially cache content using $\DD$ is unknown.  Intuitively, popular content should be seeded into the network in a way that maximizes the probability that a given $\DD$ device can find a desired file within its radio range. We explore this problem quantitively in this paper.
 
Several aspects of content caching have been studied. 
The gain offered by local caching is analyzed \cite{MaddahAli2013Journal}. Scaling laws with $\DD$ content caching are studied \cite{Ji2014}.
Optimal collaboration distance and scaling for the number of active $\DD$ links are investigated \cite{Golrezaei2014}. A distributed caching system using mobiles and deterministically placed helpers with low-rate backhaul is proposed \cite{Shanmugam2013}. Using ${\rm PPP}$s to model the user locations, optimal geographic content placement 
for various wireless network scenarios are studied \cite{Blaszczyszyn2014}. Maximum probability that the typical user finds the content in one of its nearby base stations is evaluated using the coverage number distribution \cite{Keeler2013}.

Temporal caching models have also been studied \cite{Che2002}, e.g., least recently used, least-frequently used and most recently used cache update algorithms. However, to the best of authors' knowledge, \cite{Blaszczyszyn2014} is the only work to propose a spatial caching model and there is no spatially joint content placement strategy in the literature. We aim to maximize the cache hit probability for a $\DD$ network where the spatial distribution of nodes matters, which can be exploited for an efficient use of caches. 

We consider a $\DD$ caching model in which user locations are modeled by the Poisson point process ($\PPP$) as in \cite{Lin2013}, and users have limited communication range and finite storage. $\DD$ users are primarily served by each other if  the desired content is cached in a nearby user. Otherwise, they are served by the cellular network. We aim to optimize the cache hit probability, i.e., the probability that a user can get the desired content from one of the $\DD$ users within its range. 

We propose 2 different strategies to maximize the $\DD$ cache hit probability: (i) independent content placement 
where there is no spatial correlation among users and (ii) correlated placement strategies that enable spatial diversity, namely a spatially exchangeable placement model and a Mat\'{e}rn hard core ($\mhc$) model to prioritize the caches for content placement. In the $\mhc$ model, the caches storing a particular file are never closer to each other than some given distance, so neighboring users are less likely to cache redundant content.   We show that exchangeable placement yields positively correlated spatial distribution of content, and is suboptimal in terms of the cache hit probability compared to independent placement. On the other hand, $\mhc$ placement yields a negatively correlated spatial content distribution, and a higher cache hit probability than the independent placement in the small cache size regime.

\section{System Model}
\label{model}
The locations of the $\DD$ users are modeled by the $\PPP$ $\Phi$ with density $\lambda$. We assume that there are $M$ files in total in the network and each user has a cache of the same finite size $N<M$. Depending on its cache state, each user makes requests for new files based on a general popularity distribution over the set of the files. The popularity of such requests is modeled by the Zipf distribution, which has pmf $p_r(n)=\frac{1}{n^{\gamma_r}}/\sum_{m=1}^M{\frac{1}{m^{\gamma_r}}}$, for $n=1,\hdots, M$, where $\gamma_r$ is the Zipf exponent that determines the skewness of the distribution.

$\DD$ 
users can only communicate within a finite range, which we call $\DD$ radius and denote it by $\Rdd$. A request needs to be fulfilled by the $\DD$ users within the $\DD$ radius. Otherwise, the $\DD$ user has to be served by the cellular network. The coverage process of the proposed model can be represented by a Boolean model as described next.
 
\begin{defi}
The Boolean model (BM) is based on a PPP, whose points are also called germs, and on an independent sequence of iid compact sets called grains, defined as a model driven by an independently marked PPP on $\mathbb{R}^2$ \cite{BaccelliBook1}. 
\end{defi}

Consider a given realization $\phi=\{x_i\}\subset \mathbb{R}^2$ of the $\PPP$ $\Phi$. We can think of $\phi$ as a counting measure or a point measure $\phi=\sum\nolimits_{i}{\delta_{x_i}}, \, x_i\in \mathbb{R}^2$, where $x_i$ denotes the coordinates of the $i^{\rm th}$ user and $\delta_x=\{0,1\}$ is the Dirac measure at $x$; for $A\subset \mathbb{R}^2$, $\delta_x(A)=1$ if $x\in A$ and $\delta_x(A)=0$ if $x\notin A$. Consequently, $\phi(A)$ gives the number of points of $\phi$ in $A$.

Our model is a simple $\BM$ where $x_i$'s denote the germs and $B_i(\Rdd)$ -a closed ball of radius $\Rdd$ centered at $x_i$- denote the grains. 
Then, the coverage process is driven by the following independently marked $\PPP$: 
$\tilde{\Phi}=\sum\nolimits_{i}{\delta_{(x_i,B_i(\Rdd))}}$.

The $\BM$ is given by the union $V_{\rm BM}=\bigcup_i{(x_i+B_0(\Rdd))}$ that models the coverage process of the $\DD$ transmitters. 

\begin{defi}{\bf Volume fraction \cite{BaccelliBook1}.} 
Since our model is translation invariant, volume fraction can 
be expressed as the probability that 
the origin is covered by 
$B_0(\Rdd)$ given by 
\begin{align}
\label{volfrac}
p=\mathbb{P}(0\in B_0(\Rdd))=1-\exp(-\lambda \pi\Rdds).
\end{align}
\end{defi}

We propose different strategies to serve the $\DD$ requests that maximize the cache hit. Assuming a transmitter receives one request at a time and multiple transmitters can potentially serve a request, the selection of the active transmitters depends on the caching strategies detailed in Sects. \ref{cachemodel}, \ref{exchangeable} and \ref{hardcore}.

\section{Cache Hit Probability}
\label{hitprobability}
To characterize the successful transmission probability, one needs to know the number of users that a typical node can connect to, i.e., the coverage number. Exploiting the properties of the $\PPP$, the distribution of the number of transmitters covering the typical receiver that requests file $m$ is given by 
\begin{eqnarray}
\label{pmf}
\mathcal{N}_m\sim\Pois(\lambda_m\pi {\Rdds}). 
\end{eqnarray}

Assume that the files are cached at the $\DD$ users identically and independently of each other and let $p_c(\cdot)$ be the caching probability. Let $Y_m$ be the indicator random variable that 
takes the value $1$ if file $m$ is available in the cache and $0$ otherwise. Thus, any cache satisfies the condition $\sum\nolimits_{m=1}^M{Y_m}\leq N$, i.e., $Y_m$'s are inherently dependent. However, for tractability reasons and due to the independent content placement assumption, we take the expectation of this relation and obtain our cache constraint: $\sum\nolimits_{m=1}^M{\mathbb{P}(Y_m=1)}=\sum\nolimits_{m=1}^M{p_c(m)}\leq N$.

The maximum total cache hit probability, i.e., the probability that the typical user finds the content in one of the $\DD$ users it is covered by, can be evaluated by solving 
\begin{eqnarray}
\begin{aligned}
\label{eq:hitprob-opt}
\max_{p_c} &\,\,\, \Phit\\
\textrm{s.t.}
& \quad \sum\nolimits_{m=1}^M{p_c(m)}\leq N,
\end{aligned}
\end{eqnarray}
where $\Phit=1-\sum\limits_{m=1}^M{p_r(m)\sum\limits_{k=0}^{\infty}{\mathbb{P}(\mathcal{N}_m=k)(1-p_c(m))^k}}$.

Optimal content placement is a binary problem satisfying $\sum\nolimits_{m=1}^M{Y_m}=N$. However, as noted above, the constraint in (\ref{eq:hitprob-opt}) is based on the average values of $Y_m$'s, which yields a relaxed content placement. Later, we show there are feasible solutions to the relaxed problem filling up all the cache slots.

The key step in evaluating (\ref{eq:hitprob-opt}) is to determine the coverage number distribution, i.e., $\mathbb{P}(\mathcal{N}_m=k)$. We can optimize $\Phit$ by using the Lagrangian technique as follows
\begin{eqnarray}
\mathcal{L}(\mu)
=1-\sum\nolimits_{m=1}^M{p_r(m)}\sum\nolimits_{k=0}^{\infty}{\mathbb{P}(\mathcal{N}_m=k)(1-p_c(m))^k}\nonumber\\
-\mu \big(\sum\nolimits_{m=1}^M{p_c(m)}-N\big).\nonumber
\end{eqnarray}
Taking the derivative of $\mathcal{L}(\mu)$ with respect to $p_c(m)$ and evaluating at $\mu=\mu^*$, we have $\frac{d \mathcal{L}(\mu)}{d p_c(m)}\vert_{\mu=\mu^*}=0$, for which there exists a feasible solution $p^*_c(m)$ that satisfies
\begin{eqnarray}
\label{numericalInvariant}
p_r(m)\sum\nolimits_{k=1}^{\infty}{k\mathbb{P}(\mathcal{N}_m=k){(1-p^*_{c}(m))}^{k-1}}=\mu^*,\nonumber\\ p_r(m)\mathbb{P}(\mathcal{N}_m=1)\leq\mu^*\leq p_r(m)\mathbb{E}[\mathcal{N}_m].
\end{eqnarray}
Similar to the approach in \cite{Blaszczyszyn2014}, we can use bisection method\footnote{The bisection method is a numerical root-finding method that repeatedly bisects an interval and selects a subinterval in which a root must lie. The algorithm stops when the change in the root is smaller than a chosen $\varepsilon> 0$.} and numerically solve (\ref{numericalInvariant}) to find the $p^*_c(m)$ values. We initialize the bisection method by setting $\mu$ such that $\mu\in[\mu_{\min},\mu_{\max}]$, where $\mu_{\max}=p_r(N/c_b)\mathbb{P}(\mathcal{N}_{N/c_b}=1)$ assuming $p_c(m)=1$ for $m\leq N/c_b$, hence $\mu^*\leq \mu_{\max}$, and $\mu_{\min}=p_r(c_bN)\mathbb{E}[\mathcal{N}_{c_bN}]$ assuming $p_c(m)=0$ for $m\geq c_bN$, hence $\mu^*\geq \mu_{\min}$. Here, $c_b$ is a constant integer parameter appropriately adjusted and $N$ is divisible by $c_b$ and $c_bN\leq M$.

Using the coverage number pmf (\ref{pmf}), we can rewrite (\ref{numericalInvariant}) as
\begin{eqnarray}
\mu^*=p_r(m)\sum\limits_{k=1}^{\infty}k e^{-\lambda_m\pi {\Rdds}}
\frac{(\lambda_m\pi {\Rdds})^k}{k!}{(1-p^*_{c}(m))}^{k-1}\nonumber\\
=p_r(m)\lambda_m\pi {\Rdds}\exp(-p^*_{c}(m)\lambda_m\pi {\Rdds}),\nonumber
\end{eqnarray}
which yields for $p_r(m)\mathbb{P}(\mathcal{N}_m=1)\leq\mu^*\leq p_r(m)\mathbb{E}[\mathcal{N}_m]$:
\begin{align}
\label{optimalcachingpmf}
p^*_{c}(m)=\frac{1}{\lambda_m\pi {\Rdds}}\log\Big(\frac{p_r(m)\lambda_m\pi {\Rdds}}{\mu^*}\Big).
\end{align}

\section{Independent Cache Design}
\label{cachemodel}
Given that each cache can store $N<M$ files\footnote{Swapping the contents within a cache does not change cache's state.}, our objective is to determine the number of files $L$ that should be stored in the cache with probability 1, and the maximum number of distinct files $K$ that can be stored in the cache as a function of the important design parameters, e.g., $\Rdd$, $\lambda_m$'s and $N$. Using the optimal solution $p^*_{c}(m)$ in (\ref{optimalcachingpmf}), we can deduce that
\begin{eqnarray}
\label{pcoptimal}
p^*_{c}(m)=\begin{cases}
1\quad \mu^*\leq p_r(m)\mathbb{P}(\mathcal{N}_m=1) \\
\frac{1}{\lambda_m\pi {\Rdds}}\log\big(\frac{p_r(m)\lambda_m\pi {\Rdds}}{\mu^*}\big) \, \mu^* \in \mathcal{M}_m\\
0 \quad \mu^*\geq p_r(m)\mathbb{E}[\mathcal{N}_m]
\end{cases},
\end{eqnarray}
where $\mathbb{P}(\mathcal{N}_m=1)=e^{-\lambda_m\pi {\Rdds}}(\lambda_m\pi {\Rdds})$, $\mathbb{E}[\mathcal{N}_m]=\lambda_m\pi\Rdds$ and $\mathcal{M}_m$ is a set such that for any $\mu^* \in \mathcal{M}_m$, it is satisfied that $p_r(m)\mathbb{P}(\mathcal{N}_m=1)\leq\mu^*\leq p_r(m)\mathbb{E}[\mathcal{N}_m]$. Incorporating the finite cache size constraint to (\ref{pcoptimal}), we can rewrite $\sum\nolimits_{m=1}^M{p_c(m)}$ as follows:
\begin{eqnarray}
\label{KL1}
L-1+\sum\limits_{m=L}^{K}{\frac{1}{\lambda_m\pi {\Rdds}}\log\Big(\frac{p_r(m)\lambda_m\pi {\Rdds}}{\mu^*}\Big)}=N.
\end{eqnarray}
Using the boundary conditions for $\mu^*$, we have 
\begin{eqnarray}
p_r(K)\lambda_K\pi\Rdds\leq \mu^*\leq p_r(L)e^{-\lambda_L\pi {\Rdds}}(\lambda_L\pi {\Rdds}),
\end{eqnarray}
where the relation between $L$ and $K$ can be found as
\begin{eqnarray}
\label{function1}
p_r(K)^2\leq p_r(L)^2\exp(-\lambda_L\pi {\Rdds}),
\end{eqnarray}
which follows from $\lambda_m=\lambda p_r(m)$, i.e., the density of the transmitting users is proportional to the density of the requests.

Using (\ref{pcoptimal}), for any $L\leq m\leq K$, the optimal solution is 
\begin{eqnarray}
\label{KL2}
p^*_{c}(m)=\frac{\sum\limits_{j=1}^{M}({2\gamma_r}/{j^{\gamma_r}})}{\lambda\pi \Rdds}\log\Big(\frac{K}{m}\Big)m^{\gamma_r}
+\Big(\frac{m}{K}\Big)^{\gamma_r}p_c(K).
\end{eqnarray}

From (\ref{KL1}) and (\ref{KL2}), we obtain the following relation:
\begin{multline}
\label{function2}
N-L+1
=\left[\frac{\sum\nolimits_{j=1}^{M}({2\gamma_r}/{j^{\gamma_r}})}{\lambda\pi \Rdds}\log(K)+\frac{p_c(K)}{K^{\gamma_r}}\right]\times \\
\sum\nolimits_{m=L}^{K}{m^{\gamma_r}}
-\frac{\sum\nolimits_{j=1}^{M}({2\gamma_r}/{j^{\gamma_r}})}{\lambda\pi \Rdds}\sum\nolimits_{m=L}^{K}{\log(m)m^{\gamma_r}}.
\end{multline}

Applying (\ref{function1}) with equality and from (\ref{function2}), we uniquely determine $L$ and $K$ that approximate the optimal content placement pmf in (\ref{pcoptimal}) as the following linear model:
\begin{align}
\label{pcapproximate}
p^{\rm Lin}_{c}(m)=\begin{cases}
\min\{1,1-\frac{m-L}{K-L}\} \quad 1\leq m\leq K\\
0\quad m>K
\end{cases},
\end{align}
which is a good approximation as shown in Sect. \ref{comparisonindependentmhc}.

\section{A Spatially Exchangeable Cache Model}
\label{exchangeable}
For an ordered set of $n$ transmitters covering a receiver with desired content $m$, the binary sequence $Y_{m_1}, Y_{m_2}, \dots, Y_{m_n}$ denotes the availability of the content. We assume the sequence $Y_{m_1}, Y_{m_2}, \dots, Y_{m_n}$ is \textit{exchangeable} in the spatial domain. 
\begin{defi}
An exchangeable sequence $Y_1, Y_2, Y_3, \dots$ of random variables is such that for any finite permutation $r$ of the indices $1, 2, 3, \dots$, the joint probability distribution of the permuted sequence $Y_{r(1)}, Y_{r(2)}, Y_{r(3)}, \dots$ is the same as the joint probability distribution of the original sequence.
\end{defi}

\begin{theo}\label{deFinetti}
{\bf de Finetti's theorem.}  A binary sequence $Y_1, \hdots, Y_n, \hdots$ is exchangeable if and only if there exists a distribution function $F$ on $[0, 1]$ such that for all $n$ $p(y_1, \hdots, y_n)=\int\nolimits_{0}^{1}{ \theta^{t_n}(1-\theta)^{n-t_n}\, \mathrm{d}F(\theta)}$, where $p(y_1, \hdots, y_n)=\mathbb{P}(Y_1=y_1, \hdots, Y_n=y_n)$ and $t_n=\sum\nolimits_{i=1}^n{y_i}$. It further holds that $F$ is the distribution function of the limiting frequency, i.e., if $X=\lim_{n\to\infty} \sum\nolimits_i{Y_i/n}$, then $\mathbb{P}(X\leq x)=F(x)$ and by conditioning with $X=\theta$, we obtain
\begin{eqnarray}
\mathbb{P}(Y_1 = y_1, \dots, Y_n = y_n \vert X = \theta) = \theta^{t_n}(1-\theta)^{n-t_n}.
\end{eqnarray}
\end{theo}
The optimization formulation to maximize the cache hit for an exchangeable content placement strategy becomes
\begin{eqnarray}
\begin{aligned}
\label{eq:hitprob-exchangeable}
\max_{f_{X_m}} &\,\, 
1-\sum\nolimits_{m=1}^M{p_r(m)\sum\nolimits_{k=0}^{\infty}{\mathbb{P}(\mathcal{N}_m=k)P_{\rm miss}(m,k)}}\\
\textrm{s.t.}
& \quad \sum\nolimits_{m=1}^M{\mathbb{E}[X_m]}\leq N.
\end{aligned}
\end{eqnarray}
From Theorem \ref{deFinetti}, $P_{\rm miss}(m,k)=\int\nolimits_{0}^{1}{ (1-\theta)^kf_{X_m}(\theta)\, \mathrm{d}\theta}$ is the probability that $k$ caches cover a receiver, and none has file $m$, and $\mathbb{E}[X_m]=\int\nolimits_{0}^{1}{ \theta f_{X_m}(\theta)\, \mathrm{d}\theta}$ is the probability a cache contains file $m$. Hence, the objective in (\ref{eq:hitprob-exchangeable}) is equal to 
\begin{multline}
\label{eq:hitprob-exchangeable-simplified}
\Phit=\sum\nolimits_{m=1}^M p_r(m)\Big(\int\nolimits_{0}^{1} \big(1-\sum\nolimits_{k=0}^{\infty}\exp(-\lambda_m\pi {\Rdds})\\
((\lambda_m\pi {\Rdds})^k/{k!})(1-\theta)^k\Big)f_{X_m}(\theta)\, \mathrm{d}\theta\big)\\
=1-\sum\nolimits_{m=1}^M{p_r(m)\mathbb{E}[\exp(-\lambda_m\pi {\Rdds}X_m)]}.
\end{multline}

\begin{figure*}[t!]
\centering
\includegraphics[width=0.76\textwidth]{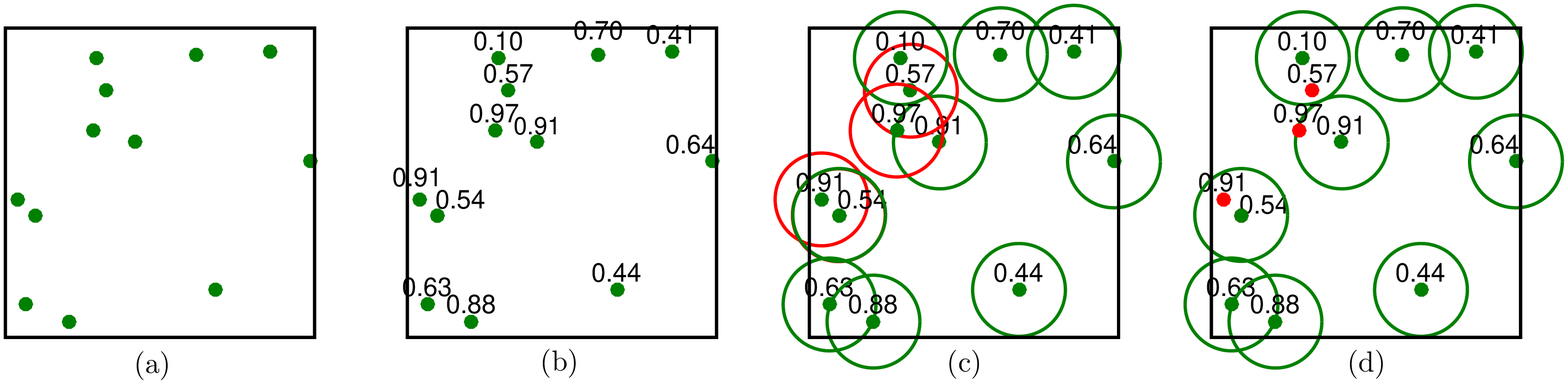}\\
\caption{\small{$\mhc$ p.p. realization: (a) Begin with a $\PPP$. (b) Associate a mark $\sim U[0,1]$ to each point independently. (c) A node $x$ is selected if it has the lowest mark among all the points in $B(x,\Rdd)$. (d) Set of selected points.}}
\label{Matern}
\end{figure*}

\begin{prop}
Any exchangeable placement strategy is worse than independent placement in terms of cache hit probability. 
\end{prop}

\begin{proof}
Using the convexity of exponential, 
we can show that the hit probability of exchangeable placement in (\ref{eq:hitprob-exchangeable-simplified}) satisfies:
\begin{eqnarray}
\label{inequalityofmethods}
1-\Phit=\sum\nolimits_{m=1}^M{p_r(m)\mathbb{E}[\exp(-\lambda_m\pi {\Rdds}X_m)]}\nonumber\\
\geq\sum\nolimits_{m=1}^M{p_r(m)\exp(-\lambda_m\pi {\Rdds}\mathbb{E}[X_m])}.
\end{eqnarray}
%
%
From (\ref{inequalityofmethods}), the miss probability of the exchangeable cache placement model is higher than the miss probability of the independent placement. 
\end{proof}

\begin{prop} 
Negatively correlated placement performs better than independent placement in terms of the hit probability. 
\end{prop}

\begin{proof}
Note that for negatively correlated content placement, i.e., when $P_{\rm miss}(m,k)\leq \mathbb{P}(Y_m=0)^k$, 
\begin{eqnarray}
&\Phit=1-\sum\nolimits_{m=1}^M{p_r(m)\sum\nolimits_{k=0}^{\infty}{\mathbb{P}(\mathcal{N}_m=k)P_{\rm miss}(m,k)}}\nonumber\\
&\geq 1-\sum\nolimits_{m=1}^M{p_r(m)\sum\nolimits_{k=0}^{\infty}{e^{-\lambda_m\pi {\Rdds}}\frac{(\lambda_m\pi {\Rdds}\mathbb{P}(Y_m=0))^k}{k!}}},\nonumber
\end{eqnarray}
which is the hit probability for independent placement. 
\end{proof}

Negatively correlated spatial placement corresponds to a distance-dependent thinning of the transmitter 
process so that neighboring users 
are less likely to 
have matching contents.

\section{Mat\'{e}rn Hard Core ($\mhc$) Content Placement}
\label{hardcore}
We propose a content placement approach exploiting the spatial properties of Mat\'{e}rn's hard core ($\mhc$) model. $\mhc$ is constructed from the underlying $\PPP$ modeling the locations of the caches by removing certain points depending on the positions of the neighboring points and additional marks attached to the points. Each transmitter of the $\BM$ $V_{\rm BM}$ is assigned a uniformly distributed mark $U[0,1]$. A node $x\in\tilde{\Phi}$ is selected if it has the lowest mark among all the points in $B(x,\Rdd)$. A realization of the $\mhc$ p.p. is illustrated in Fig. \ref{Matern}. The proposed placement model is slightly different. Instead, for each file type, there is a distinct exclusion radius. 

We optimize the exclusion radii to maximize the total hit. The exclusion radius of a particular file ($r_m$ for file $m$) depends on the file popularity in the network, transmitter density and the cache size and satisfies $r_m<\Rdd$. Otherwise, once $r_m$ exceeds $\Rdd$, as holes would start to open up in the coverage for that content, the hit probability for file $m$ would suffer. We consider the following cases: (i) if the file is extremely popular, then many transmitters should simultaneously cache the file, yielding a small exclusion radius, and (ii) if the file is not popular, then fewer transmitters would be sufficient for caching the file, yielding a larger exclusion radius. Therefore, intuitively, we might expect the exclusion radius to decrease with increasing file popularity. However, our analysis shows that the exclusion radius is positively correlated with the file popularity, i.e., the most popular files are stored in a few caches with higher marginal probabilities unlike the files with low popularity that are stored with lower marginals.

Given the exclusion radius of the $\mhc$ model, a file should be placed at only one cache within a circular region. Hence, the caching probability of file $m$ at a typical transmitter is
\begin{eqnarray}
\label{cacheprobmatern}
p_{\rm cache}(m)\overset{(a)}{=}\mathbb{E}\Big[\frac{1}{1+C_m}\Big]=\frac{1-\exp(-\bar{C}_m)}{\bar{C}_m},
\end{eqnarray}
where $C_m$ is number of neighboring transmitters in a circular region of radius $r_m$, which is Poisson distributed with mean $\bar{C}_m=\lambda\pi r_m^2$ as $C_m\sim\Poisson(\bar{C}_m)$, and $(a)$ follows from the fact that the caching probability of a typical transmitter at origin is given by the probability that the node qualifies and gets the minimum mark value in its neighborhood. 

\begin{figure*}[t!]
\begin{minipage}[t]{.32\textwidth}
\centering
\includegraphics[width=\textwidth]{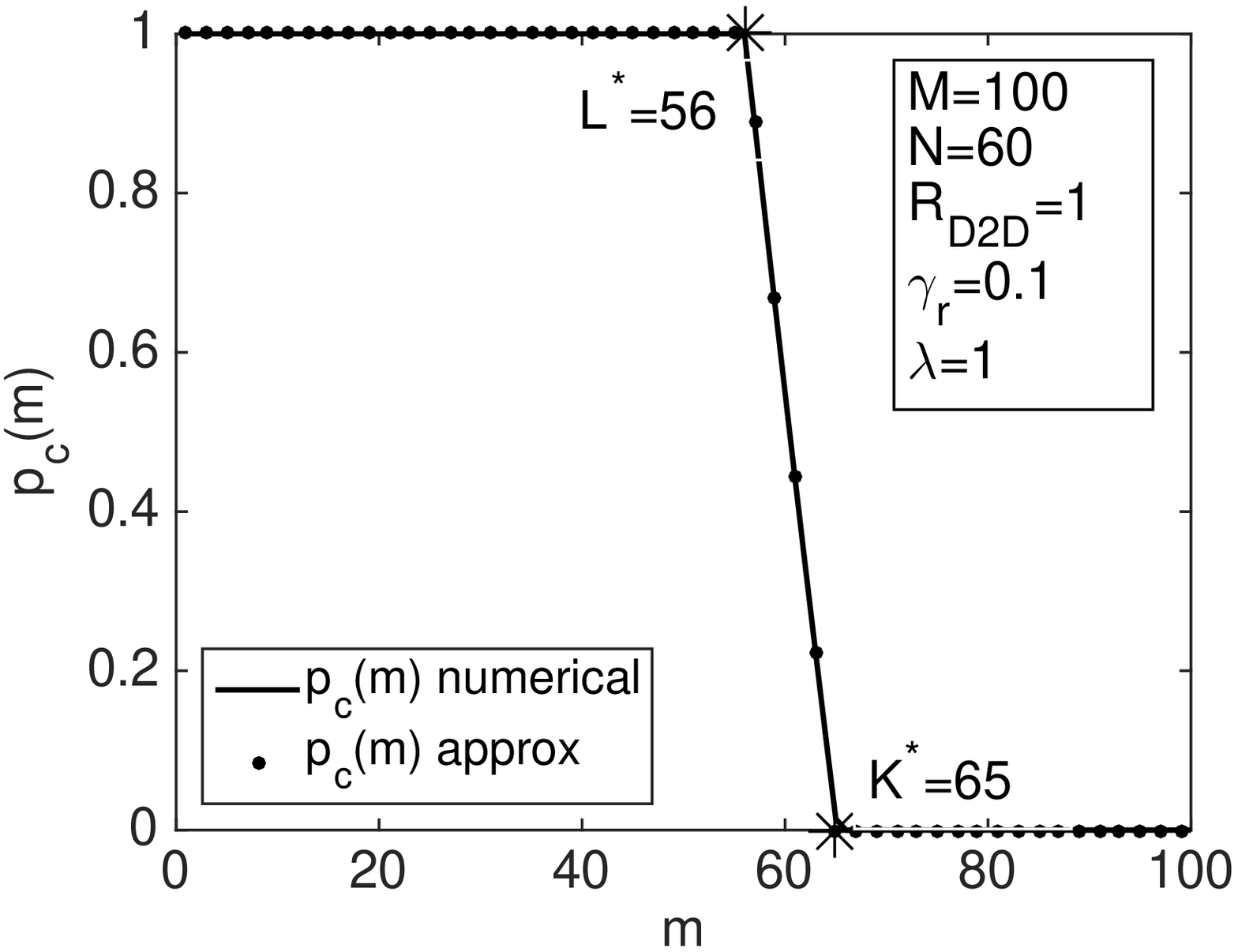}
\end{minipage}
\hfill
\begin{minipage}[t]{.32\textwidth}
\centering
\includegraphics[width=\textwidth]{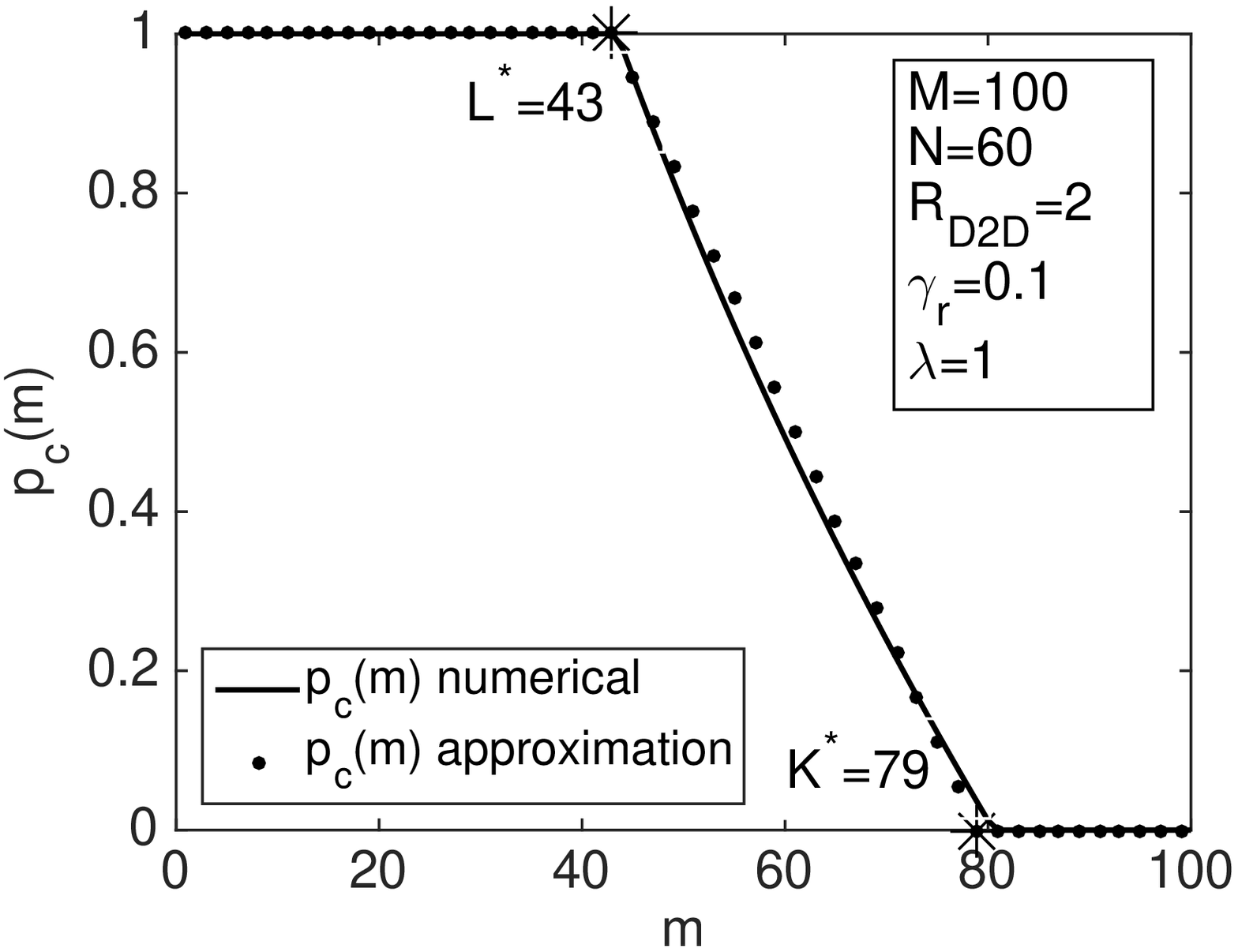}
\end{minipage}
\hfill
\begin{minipage}[t]{.32\textwidth}
\centering
\includegraphics[width=\textwidth]{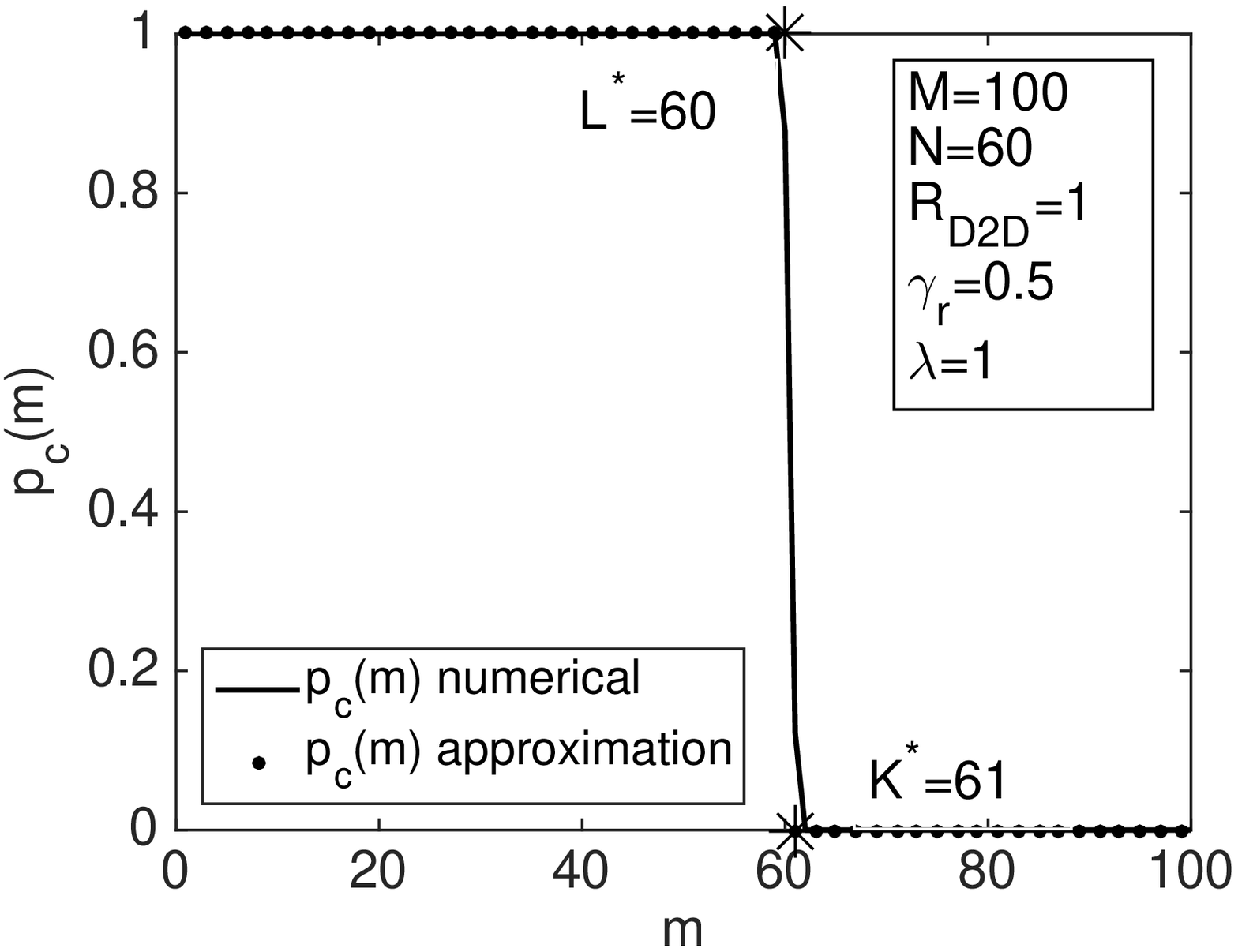}
\end{minipage}
\hfill
\caption{\small{Optimal cache placement (independently at each user) with more focused content popularity.}\label{OptCachingProb}}
\end{figure*}

Let $\tilde{C}_m$ be the number of transmitters containing file $m$ within a circular region of radius $r_m$. Since only one transmitter is allowed to contain a file within the exclusion radius, $\tilde{C}_m\in\{0,1\}$. Given the $\mhc$ model as above, there exists $\tilde{C}_m\in\{0,1\}$ transmitter having file $m$ in a region of size $\pi r_m^2$ with the probability of having one transmitter having file $m$ is 
\begin{eqnarray}
\label{probhavingoneTX}
\mathbb{P}(\tilde{C}_m=1)=1-\exp(-\bar{C}_m).
\end{eqnarray}
Hence, $\mathbb{E}[\tilde{C}_m]=\lambda_{\mhc}(m)\pi r_m^2=1-\exp(-\bar{C}_m)$ \cite[Ch. 2.1]{BaccelliBook1}, where $\lambda_{\mhc}(m)$ is the density of the $\mhc$ model. 

Consider the following optimization formulation:
\begin{equation}
\begin{aligned}
\label{eq:hitprob-matern}
\max_{p_{\rm cache}(m)} &\quad \Phit=\sum\nolimits_{m=1}^M{p_r(m)\mathbb{P}(\tilde{C}_m=1)}\\
\textrm{s.t.}
& \quad \sum\nolimits_{m=1}^M{p_{\rm cache}(m)}\leq N,
\end{aligned}
\end{equation}
which gives the maximum hit for the $\mhc$ content placement, where the constraint equation follows from that the probability that content $m$ is cached at a transmitter is equal to $p_{\rm cache}(m)$, and there are at most $N$ files to be stored in each cache.

We define the Lagrangian to find the solution as follows: $\mathcal{M}(\zeta)=\sum\nolimits_{m=1}^M{p_r(m)(1-e^{-\bar{C}_m})}+\zeta \big(\sum\nolimits_{m=1}^M{\frac{1-e^{-\bar{C}_m}}{\bar{C}_m}}-N\big)$, and taking its derivative with respect to $\bar{C}_m$, we get
\begin{align}
\frac{d \mathcal{M}(\zeta)}{d \bar{C}_m}
=p_r(m)e^{-\bar{C}_m}
+\zeta\Big(\frac{e^{-\bar{C}_m}\bar{C}_m-(1-e^{-\bar{C}_m})}{\bar{C}_m^2}\Big).\nonumber
\end{align}
Evaluating this at $\zeta=\zeta^*$, we obtain the simplified relation $p_r(m)\bar{C}_m^2+\zeta^*\left(\bar{C}_m-(\exp(\bar{C}_m)-1)\right)=0$, or equivalently, 
\begin{align}
\label{zetaopt}
\zeta^*=h_m(\bar{C}_m^*)={p_r(m)(\bar{C}_m^*)^2}/{(\exp(\bar{C}_m^*)-\bar{C}_m^*-1)},
\end{align} 
where 
the optimal solution $\zeta^*$ depends on the cache size $N$. Note that $\zeta^*$ is decreasing in $\bar{C}_m$, $\lim_{\bar{C}_m\to 0} \,\zeta^*=2p_r(m)$ and $\lim_{\bar{C}_m\to \infty} \,\zeta^*=0$. 
We 
determine the optimal value of $\bar{C}_m$ as 
\begin{eqnarray}
\bar{C}_m^*=\begin{cases}
0 \quad \text{if} \quad \zeta^*\geq 2p_r(m),\\
h_m^{-1}(\zeta^*) \quad \text{if} \quad \zeta^*< 2p_r(m).
\end{cases}
\end{eqnarray}

For very unpopular files with small $p_r(m)$, $\zeta^*$ satisfies $\zeta^*>2p_r(m)$ and hence, $\bar{C}_m^*=0$ and as the file popularity increases, $p_r(m)$ will be higher and $\zeta^*$ satisfies the relation $\zeta^*\leq 2p_r(m)$. Hence $\bar{C}_m^*$ increases with popularity and satisfies the relation $h_m^{-1}(\zeta^*)$. Thus, the average number of transmitters within the exclusion region, i.e., $\bar{C}_m^*$, is increasing by increasing the file popularity, and the exclusion radius for files with high popularity should be higher, which yields lower $p_{\rm cache}(\cdot)$ for popular files from (\ref{cacheprobmatern}).


As the storage size of the users 
drops, the exclusion region should increase to bring more spatial diversity into the model. Using the constraint in (\ref{eq:hitprob-matern}), as $N$ drops, a typical receiver won't be able to find its requested files and $\lim_{N\to 0} \,r_m=\infty$, which increases the volume fraction, i.e., increases the caching probability. When $N$ increases sufficiently, the transmitter candidates of the typical receiver will have any of the requested files and $\lim_{N\to \infty} \,r_m=0$, and because it is redundant to cache the files at all the transmitters, the exclusion radius should be made smaller to decrease the volume fraction and the caching probability. Thus, $N$ and $r_m$ have an inverse relationship.

\section{A Comparison of Content Placement Models}
\label{comparisonindependentmhc}
We compare the optimal solution $p^*_{c}(m)$ (\ref{optimalcachingpmf}) and our linear approximation (\ref{pcapproximate}) in Fig. \ref{OptCachingProb}. Modifying the $\DD$ parameters, we observe that our linear solution in (\ref{pcapproximate})  is indeed a good approximation of the optimal solution in (\ref{pcoptimal}). Keeping $\gamma_r$ constant, by increasing $\Rdd$, we expect to see a more diverse set of requests from the user, $L$ to decrease and $K$ to increase. Converse is also true. When we keep $\Rdd$ fixed, and increase $\gamma_r$, since the requests become more skewed towards the most popular files, the optimal strategy for the user is to store the most popular files in its cache. Keeping $\Rdd$ and $\gamma_r$ fixed, and increasing $\lambda$ has a similar effect as increasing $\Rdd$, however 
this is not plotted 
due to space limitations. From these plots, although it is clear that independent placement favors the most popular contents, it is not optimal to cache the most popular contents everywhere.

The performance of the independent content placement and the $\mhc$-based model is mainly determined by the cache size. Hence, the analysis boils down to finding the critical cache size that determines which model outperforms the other in terms of the hit probability under or above the critical size. Using the hit probabilities given in (\ref{eq:hitprob-opt}) and (\ref{eq:hitprob-matern}), respectively for the independent and $\mhc$ content placements, the required condition for which the $\mhc$ model performs better than the independent placement is $\sum\nolimits_{m=1}^M{p_r(m)\mathbb{P}(\tilde{C}_m=1)}\geq 1-\sum\nolimits_{m=1}^M{p_r(m)\sum\nolimits_{k=0}^{\infty}{\mathbb{P}(\mathcal{N}_m=k)(1-p_c(m))^k}}$. A sufficient condition for this to be valid is given as
\begin{eqnarray}
\label{condMHC}
\mathbb{P}(\tilde{C}_m=1)\geq 1-\sum\nolimits_{k=0}^{\infty}{\mathbb{P}(\mathcal{N}_m=k)(1-p_c(m))^k},
\end{eqnarray} 
equivalent to the condition $e^{-\lambda_m\pi r_m^2}\leq 
e^{-p_c(m)\lambda_m\pi {\Rdds}}$.

Now, we consider two regimes controlled by the cache size $N$. In the regime where $\mhc$ placement is better than the independent placement, using (\ref{condMHC}), $r_m$ is lower bounded as $\sqrt{p_c(m)}\Rdd \leq  r_m$, for all $m$, and the volume fraction is lower bounded by $1-\exp(-\lambda_m \pi p_c(m)\Rdds)$. Since a high exclusion radius is required for small cache size, $\mhc$ placement performs better than the independent placement for small cache size. When $r_m<\sqrt{p_c(m)}\Rdd$, the volume fraction is upper bounded by $1-\exp(-\lambda_m \pi p_c(m)\Rdds)$. In this case, the file exclusion radii are very small for files with very low popularity, implying that the cache size should be sufficiently large, for which case independent placement is better than $\mhc$ placement. The cache hit probability trends of the independent placement in \cite{Blaszczyszyn2014}, and the $\mhc$ placement model with respect to the cache size are shown in Fig. \ref{fig-hit-MHC}.
\begin{figure}[t!]
\centering
\includegraphics[width=0.32\textwidth]{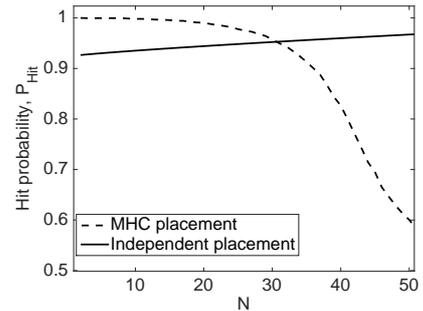}
\caption{\small{Cache hit probabilities of the independent and $\mhc$ models.}\label{fig-hit-MHC}}
\end{figure}

\bibliographystyle{IEEEtran}
\bibliography{D2Dreferences}

\end{document}